\documentclass[12pt,numbers]{elsarticle}

\usepackage{amssymb}
\usepackage{amsmath}
\usepackage{pslatex}
\usepackage{amsfonts,color,morefloats}
\usepackage{amssymb,amsmath,latexsym,amsthm}
\usepackage{hyperref}
\usepackage{multirow}
\usepackage{makecell}
\usepackage{booktabs}
\usepackage{booktabs}
\usepackage{graphicx}
\hypersetup{hidelinks}
\allowdisplaybreaks

\newcommand{\gf}{{\mathbb{F}}}

\newtheorem{theorem}{Theorem}
\newtheorem{lemma}{Lemma}

\newtheorem{definition}{Definition}
\newtheorem{example}{Example}

\newtheorem{remark}{Remark}

\begin{document}

\begin{frontmatter}

\title{A class of locally differentially $4$-uniform power functions with Niho exponents}


%\author[label1]{Ketong Ren} %% Author name
%\ead{rkt@my.swjtu.edu.cn}
%\author[label2]{Maosheng Xiong} %% Author name
%\ead{mamsxiong@ust.hk}
\author[label1]{Haode Yan}
\ead{hdyan@swjtu.edu.cn}

\author[label2]{Kangquan Li\corref{cor1}}
\ead{likangquan11@nudt.edu.cn}
\cortext[cor1]{Corresponding author}

\affiliation[label1]{organization={School of Science},%Department and Organization
            addressline={Harbin Institute of Technology}, 
            city={Shenzhen},
            postcode={518055}, 
            country={China}}
\affiliation[label2]{organization={College of Science},%Department and Organization
	addressline={National University of Defense Technology}, 
	city={Changsha},
	postcode={410072}, 
	country={China}}

%\affiliation[label2]{organization={Department of Mathematics},%Department and Organization
%             addressline={The Hong Kong University of Science and Technology}, 
%             city={Hong Kong},
%             country={China}}

%% Abstract
\begin{abstract}
	Niho exponents have found important applications in sequence design, coding theory, and cryptography. Determining the differential spectrum of a power function with Niho exponent is a topic of considerable interest. In this paper, we  investigate the power function $F(x) = x^{3q - 2}$ over $\mathbb{F}_{q^2}$, where $q = 2^m$ and $m\geq 4$ is an even integer. Notably, the exponent $3q - 2$ is a Niho exponent. By analyzing the properties of certain polynomials over $\mathbb{F}_{q^2}$, we determine the differential spectrum of $F$. Our results show that $F$ is locally differentially $4$-uniform, which complements existing results on the differential spectra of power functions with Niho exponents.
\end{abstract}

%% Keywords
\begin{keyword}
 cryptographic function; power functions; differential uniformity; differential spectrum; Niho exponent 
  
  MSC: 11T06, 94A60
\end{keyword}

\end{frontmatter}

\section{Introduction}
Let $ \mathbb{F}_q $ denote the finite field with $ q $ elements, where $ q = p^n $ is a prime power and $ n $ is a positive integer. Any cryptographic function $ F: \mathbb{F}_q \to \mathbb{F}_q $ can be uniquely represented as a univariate polynomial over $ \mathbb{F}_q $ of degree less than $ q $; hence, we may regard $ F $ as an element of the polynomial ring $ \mathbb{F}_q[x] $. Such functions are widely used in the design of substitution boxes (S-boxes), which are fundamental components in symmetric cryptosystems such as block ciphers. In this context, it is crucial for cryptographic functions to resist differential cryptanalysis, a powerful attack introduced by Biham and Shamir~\cite{BS91}. To measure the resistance of a given function $ F $ to such attacks, Nyberg~\cite{N94} introduced the concept of differential uniformity, which is closely related to the difference distribution table (DDT) of $ F $. These notions are defined as follows. For any $ a, b \in \gf_{q}$, the DDT entry of $F$ at the  point $ (a,b) $, denoted $ \delta_F(a,b) $, is defined as \[\delta_F(a,b)=\big{|} \{x \in \gf_{q}: ~F(x+a)-F(x)=b\} \big{| },\]
where $ \big{|} S \big{|} $ denotes the cardinality of the set $S$.
The differential uniformity of the function $ F $, denoted by $ {\Delta _F} $, is then defined as
\[{\Delta _F}=\max\{\delta_F(a,b): ~a \in {\mathbb{F}_{q}^*}, b \in \mathbb{F}_{q}\}  ,\]
where $\gf_{q}^*=\gf_{q}\setminus\{0\}$. In this case, we say that $ F $ is differentially $\Delta_F$-uniform. For a cryptographic function $ F $, a smaller value of $ \Delta_F $ indicates a stronger resistance to differential attacks. When $ \Delta_F = 1 $, the function $ F $ is said to be perfect nonlinear (PN); and when $ \Delta_F = 2 $, it is called almost perfect nonlinear (APN). We note that perfect nonlinear (PN) functions exist only over finite fields of odd characteristic, whereas over finite fields of even characteristic the lowest possible differential uniformity is 2. PN and APN functions play a significant role in the theoretical foundations of various fields, including coding theory and combinatorics. Recent progress on PN and APN functions has been substantial, and can be found in the literature; see a recent excellent book \cite{carlet2021boolean} by Carlet or some papers, e.g. \cite{BBM08,BCCCV20,BHK20,DMMPW03,D992,D991,D01,golouglu2022biprojective,golouglu2023exponential,HRS99,li2023two,ZW11,ZKW09,ZW09}, and the references therein.

Power functions with low differential uniformity are widely used in the design of S-boxes due to their strong resistance to differential attacks and their typically low implementation cost in hardware. When $ F $ is a power function, i.e., $ F(x) = x^d $ for some integer $ d $, it is clear that
$$
\delta_F(a, b) = \delta_F(1, b / a^d)
$$
for all $ a \in \mathbb{F}_q^* $ and $ b \in \mathbb{F}_q $. Therefore, the differential properties of $ F $ are entirely determined by the values of $ \delta_F(1, b) $ as $ b $ ranges over $ \mathbb{F}_q $. This uniform behavior has attracted significant attention in the cryptographic community and has been extensively studied. The concept of the differential spectrum of a power function was formally introduced by Blondeau, Canteaut, and Charpin in \cite{BCC10}, as follows.
\begin{definition}\label{def1}
	Let $ F(x) = x^d $ be a power function over $ \mathbb{F}_q $ with differential uniformity $ \Delta_F $. Define
	$$
	\omega_i = \big| \{ b \in \mathbb{F}_q : \delta_F(1, b) = i \} \big|, \quad 0 \leq i \leq \Delta_F.
	$$
	The differential spectrum of $ F $ is defined as the multiset
	$$
	DS_F = \{ \omega_i > 0 : 0 \leq i \leq \Delta_F \}.
	$$
\end{definition}
The study of the differential spectrum, and its connections to the Walsh spectrum, dates back to Dobbertin et al.~\cite{DHKM01}. Moreover, it was shown in \cite{BCC10} that the elements of the differential spectrum satisfy the following two fundamental identities:
\begin{equation}\label{omegaiomega}
	\sum_{i=0}^{\Delta_F} \omega_i = q, \quad \text{and} \quad \sum_{i=0}^{\Delta_F} i \omega_i = q.
\end{equation}
The following lemma, which can be derived from Theorem 10 in \cite{HRS99}, plays a crucial role in determining the differential spectrum of a power function $ F(x) = x^d $.
\begin{lemma}[\cite{HRS99}]\label{i2omega}
	Let the notation be as in Definition~\ref{def1}, and let $ N_4 $ denote the number of solutions $ (x_1, x_2, x_3, x_4) \in (\mathbb{F}_q)^4 $ to the system of equations
	\begin{equation}\label{equsystem}
		\left\{
		\begin{aligned}
			x_1 - x_2 + x_3 - x_4 &= 0, \\
			x_1^d - x_2^d + x_3^d - x_4^d &= 0.
		\end{aligned}
		\right.
	\end{equation}Then we have
	\begin{equation}\label{i^2omega}
		\sum_{i=0}^{\Delta_F} i^2 \omega_i = \frac{N_4 - q^2}{q - 1}.
	\end{equation}
\end{lemma}
Determining the differential spectra of cryptographic functions, especially those that are bijective and have low differential uniformity, is of great importance in symmetric cryptography. However, computing the differential spectrum of a given power function is generally a difficult task. To date, only a limited number of classes of power functions have been fully analyzed in this regard. The known results on power functions $ F $ over $ \mathbb{F}_{2^n} $ for which the differential spectrum has been determined are summarized in Table~\ref{table}. For results on power functions over $ \mathbb{F}_{p^n} $ with $ p $ odd, the reader is referred to the works \cite{CHNC13,DHKM01,JLLQ21,JLLQ22,MXLH22,PLZ23,TY23,TC17,XZLH20,YL21,YMT23,YMT24,YX22,YDZQ25}.

For the binary case ($ p = 2 $), the values $ \delta_F(1, 0) $ and $ \delta_F(1, 1) $ are often exceptionally large and exhibit trivial or degenerate behavior. Consequently, they provide little insight into the cryptographic strength of $F$ and are typically excluded from meaningful differential analysis. This observation has led researchers to focus on the behavior of the differential spectrum for $ \delta_F(1, b) $ with $ b \notin \mathbb{F}_2 $. Motivated by this, the notion of locally APN functions was introduced in \cite{BCC11}, and the concept of locally differentially 4-uniform functions can be naturally generalized. We recall the following definition.
\begin{definition}
	Let $ F $ be a power function on $ \mathbb{F}_{2^n} $. We say that $ F $ is locally APN (respectively, locally differentially 4-uniform) if
	$$
	\delta_F(1, b) \leq 2 \quad (\text{respectively},~ \delta_F(1, b) \leq 4)
	$$
	for all $ b \in \mathbb{F}_{2^n} \setminus \mathbb{F}_2 $.
\end{definition}

Niho exponents were first introduced by Niho in 1972 \cite{N72}, who studied the cross-correlation between an $ m $-sequence and its decimated version. Since then, such exponents have found wide applications in various areas, including cryptography and coding theory. For recent progress on the use of Niho exponents, we refer the reader to \cite{LZ19,A21}.

The differential spectrum of the power function $ f(x) = x^{2q - 1} $ over $ \mathbb{F}_{q^2} $, where $ q = 2^m $ and $ 2q - 1 $ is a Niho exponent, was studied in \cite{BCC11}. It was shown that $ f $ is locally APN. In this paper, we investigate the power function $ F(x) = x^{3q - 2} $ over $ \mathbb{F}_{q^2} $, where $ q = 2^m $ and $ m\geq 4$ is an even integer. Notably, $ 3q - 2 $ is also a Niho exponent. We mention that $F$ is CCZ-inequivalent to the power functions with known differential spectra. The rest of this paper is organized as follows. In Section~\ref{sec:pre}, we recall some preliminary results, introduce the properties of a class of trinomials, and determine the number of solutions of a certain system of equations over finite fields, which will be used in the sequel. The main results are presented in Section~\ref{sec:main}, where we derive the differential spectrum of $ F $ and provide some examples. Finally, Section~\ref{sec:con} concludes the paper.

\begin{table}[!t]\label{table}
	\footnotesize
	\centering
	\caption{Power Functions over $\gf_{2^n}$ with Known Differential Spectra}
	\begin{tabular}{|c|c|c|c|}
		\hline
		$d$ & Condition & $\Delta_F$ & Ref \\
		\hline
		
		$2^t+1$ & gcd$(t,n)=s$ & $2^s$ & \cite{BCC10} \\
		\hline
		
		$2^{2t}-2^t+1$ & gcd$(t,n)=s,\frac{n}{s}\,$ odd & $2^s$ & \cite{BCC10} \\
		\hline
		
		$2^n-2$ & $n\geqslant2$ & $2$\,or\,$4$ & \cite{BCC10} \\
		\hline
		
		$2^{2k}+2^k+1$ & $n=4k$ & $4$ & \cite{BCC10},\cite{XY17} \\
		\hline
		
		$2^t-1$ & $t=3,n-2$ & $6$ or $8$ & \cite{BCC11} \\
		\hline
		
		$2^t-1$ & $t=\frac{n}{2}$, $n$ even & \makecell{$2^\frac{n}{2}-2$\\(locally APN)}  & \cite{BCC11} \\
		\hline
		$2^t-1$ & $t=\frac{n}{2}+1$, $n$ even & \makecell{$2^\frac{n}{2}$\\(locally APN)} & \cite{BCC11} \\
		\hline
		
		$2^t-1$ & $t=\frac{n-1}{2}$, $\frac{n+3}{2}$, $n$ odd & $6$ or $8$ & \cite{BP14} \\
		\hline
		
		$2^m+2^{(m+1)/2}+1$ & $n=2m$, $m\geqslant5$ odd & $8$ & \cite{XYY18} \\
		\hline
		
		$2^{m+1}+3$ & $n=2m$, $m\geqslant5$ odd & $8$ & \cite{XYY18} \\
		\hline
		
		$2^{3k}+2^{2k}+2^k-1$ & $n=4k$ & $2^{2k}$ & \cite{TLWZTJ23} \\
		\hline
		
		$k(2^m-1)$ & $n=2m$, gcd$(k,2^m+1)=1$ & \makecell{$2^m-2$\\(locally APN)} & \cite{HLXZT23} \\
		\hline
		
		$\frac{2^m-1}{2^k+1}+1$ & $n=2m$, gcd$(k,m)=1$ & \makecell{$2^m$\\(locally APN)} & \cite{XML23} \\
		\hline
		$2^m+3$ & $n=2m$, & \makecell{ $2^m$ (locally differentially $4$-uniform) \\  or $2^m+2$}  & \cite{li2023differential} \\  		
		\hline
		$3\cdot 2^{m}-2$ & $n=2m$, $m\geq 4$ even  & \makecell{$2^m$\\(locally differentially $4$-uniform)} & This paper \\
		\hline
	\end{tabular}
	\label{tab:DS}
\end{table}

\section{Preliminaries}\label{sec:pre}
In this section, we present several lemmas that will be used in the following sections.

\subsection{On the multiplicative subgroup $U_{q+1}$ of $\gf_{q^2}^*$}
Let $q=2^m$ for some positive integer $m$. We define $U_{q+1}=\{x\in \gf_{q^2}:x^{q+1}=1\}$, which is a cyclic subgroup of $\gf_{q^2}^*$ of order $q+1$. This multiplicative subgroup plays a central role in our discussion. The following lemma is well-known, and its proof is therefore omitted.

\begin{lemma}\label{xyz}
	Every element $ x \in \mathbb{F}_{q^2}^* $ can be written uniquely in the form
	\[x=yz,\]
	where $ y \in \mathbb{F}_q^* $ and $ z \in U_{q+1} $.
\end{lemma}

Moreover, the following result was established in \cite{TD21}.
\begin{lemma}\label{usum}
	Let $u_1,u_2,u_3,u_4$ be four distinct elements in $U_{q+1}$. Then 
	\[u_1+u_2+u_3+u_4\neq 0.\]
\end{lemma}

\subsection{On a class of trinomials over $\gf_{q^2}$}
In this subsection, we study the cardinalities of preimage sets for a specific class of trinomials over $\gf_{q^2}$. Specifically, we consider trinomials of the form
\[
f(x) = x^{2} + x^{1-q} + x^{2-q}.
\]
For the convenience, we define $0^{-1}=0$.

\begin{theorem}\label{image}
	Let $f(x)=x^2+x^{1-q}+x^{2-q}$ be a trinomial over $\gf_{q^2}$, where $q=2^m$ and $m$ is a positive even integer. For each $ c \in \mathbb{F}_{q^2} $, define the preimage set
	$$
	f^{-1}(c) = \{ x \in \mathbb{F}_{q^2} : f(x) = c \}.
	$$
	Then the following hold:
	\begin{itemize}
		\item[(i)]   $ |f^{-1}(0)| = 3 $;
		\item[(ii)]  $ |f^{-1}(1)| = q + 1 $; and
		\item[(iii)] $ |f^{-1}(c)| \leq 2 $ for all $ c \in \mathbb{F}_{q^2} \setminus \mathbb{F}_2 $.
	\end{itemize}
\end{theorem}

\begin{proof} (i). We consider the equation 
	\begin{equation}\label{f1}
		x^2+x^{1-q}+x^{2-q}=0.
	\end{equation}
	It is clear that $x=0$ is a solution of (\ref{f1}). Now we suppose $x\neq 0$. Then multiplying both sides of (\ref{f1}) by $x^{q-1}$, we obtain
	\begin{equation}\label{f11}
		\overline{x}x+x+1=0.
	\end{equation}
	Herein and hereafter, $\overline{x}:=x^q$ for any $x\in\gf_{q^2}$. From (\ref{f11}), it follows that
	$$x=\overline{x}x+1\in\gf_q,$$ 
	which implies $\overline{x}=x$. Substituting back into (\ref{f11}), we get
	$$x^2+x+1=0.$$
	Therefore, the non-zero solutions of \eqref{f1} are precisely the roots of the above quadratic equation, 
	which are $x=\omega$ and $x=\omega^2$, where $\omega$ is a primitive element in $\gf_4$. Since $m$ is even, we have $\omega^q=\omega$. It is then straightforward to verify that both $x=\omega$ and $x=\omega^2$ satisfy the original equation \eqref{f1}. This completes the proof of part (i).
	
	(ii). Consider the equation
	\begin{equation}\label{f2}
		x^2+x^{1-q}+x^{2-q}=1.
	\end{equation} 
	Clearly, $x=0$ is not a solution of \eqref{f2}. Multiplying both sides by $x^q$ (since $x\neq 0$), we obtain
	\begin{equation*}
		x^{q+2}+x+x^2=x^q,
	\end{equation*}
	which simplifies to 
	\begin{equation}\label{f2'}
		(x+1)(x\overline{x}+\overline{x}+x)=0.
	\end{equation}
	It is straightforward to verify that $x=1$ is a solution of \eqref{f2'}. Now suppose  $x\overline{x}+\overline{x}+x=0$. Then we have 
	\[x\overline{x}+\overline{x}+x+1=1,\]
	which can be rewritten as
	\[(x+1)^{q+1}=1.\]
	Therefore, $x+1\in U_{q+1}$, which implies that $x=u+1$ for some $u\in U_{q+1}$. Conversely, for any $u\in U_{q+1}\setminus\{1\}$, let us substitute 
	$x=u+1$ into the original equation
	\[f(u+1)=(u+1)^2+(u+1)^{1-q}+(u+1)^{2-q}\]
	Using the fact that $u^{q+1}=1$, and simplifying each term accordingly, we find
	\[f(u+1)=u^2+1+\frac{u+1}{u^{-1}+1}+\frac{u^2+1}{u^{-1}+1}=1.\]
	Hence, every element of the form $x=u+1$ with $u\in U_{q+1}\setminus\{1\}$ satisfies \eqref{f2}. Since $|U_{q+1}|=q+1$, and only $u=1$ leads to $x=0$, which is not a solution, it follows that there are exactly $q+1$ distinct solutions in total. This completes the proof of part (ii).

	(iii). Consider the equation
	\begin{equation}\label{f3}
		x^2+x^{1-q}+x^{2-q}=c
	\end{equation} 
	for  $c\in\gf_{q^2}\setminus\gf_2$. Since $x=0$ is not a solution, we may multiply both sides by $x^q$
	to obtain
	\begin{equation}\label{f31}
		(x^2+c)\overline{x}=x^2+x.
	\end{equation}
	Note that the right-hand-side of (\ref{f31}) is not zero for any solution $x$
	of \eqref{f3},  so we can solve for $\overline{x}$ as
	\begin{equation}\label{f32}
		\overline{x}=\frac{x^2+x}{x^2+c}.
	\end{equation}
	Taking the $q$-th power of both sides of  (\ref{f32}), we get
	\[x=\frac{\overline{x}^2+\overline{x}}{\overline{x}^2+\overline{c}}.\]
	Substituting \eqref{f32} into this expression yields
	\[x=\frac{(\frac{x^2+x}{x^2+c})^2+\frac{x^2+x}{x^2+c}}{(\frac{x^2+x}{x^2+c})^2+\overline{c}}.\]
	After simplification, this leads to the quartic equation
	\begin{equation}\label{f33}
		x^4+\frac{c+1}{\overline{c}+1}x+\frac{c(c\overline{c}+1)}{\overline{c}+1}=0
	\end{equation}
	Since $c\neq 1$, this quartic equation has at most four roots in $\gf
	_{q^2}$,  implying that \eqref{f31} has at most four solutions.  We now proceed to show that, in fact, it has at most two solutions.
	
	From (\ref{f33}), we deduce that
	\[x^4=\frac{c+1}{\overline{c}+1}x+\frac{c(c\overline{c}+1)}{\overline{c}+1}.\]
	Raising both sides to the fourth power gives
	\[x^{4^2}=(\frac{c+1}{\overline{c}+1})^4x^4+(\frac{c(c\overline{c}+1)}{\overline{c}+1})^4.\]
	Substituting back $x^4$ from above, we find that $x^{16}$ can also be expressed as a linear function of $x$.
	In fact, one can show inductively that for any positive integer $k$, the power $x^{4^k}$ can be written in the form 
	\[x^{4^k}=a_kx+b_k,\]
	where $a_k,b_k\in \gf_{q^2}$ depend on $c$ and $k$. Specifically, for $k=1$, we have \[a_1=\frac{c+1}{\overline{c}+1}, b_1=\frac{c(c\overline{c}+1)}{\overline{c}+1}.\]
	In general, it holds that
	\[a_k=(\frac{c+1}{\overline{c}+1})^\frac{4^k-1}{3}.\]
	In particular, since $m$ is even, we have
	\[\overline{x}=x^{4^{\frac{m}{2}}}=a_{\frac{m}{2}}x+b_{\frac{m}{2}}.\]
	Substituting this into \eqref{f31}, we arrive at the cubic equation
	\begin{equation}\label{f34}
		a_{\frac{m}{2}}x^3+(b_{\frac{m}{2}}+1)x^2+(a_\frac{m}{2}c+1)x+b_{\frac{m}{2}}c=0.
	\end{equation}
	Note that $a_{\frac{m}{2}}\neq 0$, so this is indeed a cubic equation, which has at most three solutions.
	Now, multiplying \eqref{f34} by $x$, and using $x^4=a_1x+b_1$, we obtain another relation.
	\begin{equation}\label{f35}
		(b_{\frac{m}{2}}+1)x^3+(a_\frac{m}{2}c+1)x^2+(a_1a_{\frac{m}{2}}+b_{\frac{m}{2}}c)x+a_\frac{m}{2}b_1=0.
	\end{equation}
	If $b_{\frac{m}{2}}+1=0$, then (\ref{f35}) reduces to an equation with degree at most two, which has at most two roots. Otherwise, we eliminate $x^3$
	between \eqref{f34} and \eqref{f35}, resulting in a quadratic equation
	\begin{equation}\label{f36}
		(a_{\frac{m}{2}}^2c+a_\frac{m}{2}+b^2_{\frac{m}{2}}+1)x^2+(a_1a_{\frac{m}{2}}^2+a_\frac{m}{2}c+b_\frac{m}{2}+1)x+b_1a_{\frac{m}{2}}^2+(b_{\frac{m}{2}}+1)b_{\frac{m}{2}}c=0.
	\end{equation}
	We observe that the left-hand side cannot be identically zero, because \eqref{f31} has at most four solutions. Therefore, this quadratic equation has degree at most two, and hence at most two roots.
	
	Thus, we conclude that \eqref{f3} has at most two solutions in $\gf_{q^2}$. This completes the proof of part (iii).
\end{proof}

\subsection{On the number of solutions of a certain system of equations}
In this subsection, we study the number of solutions of a certain system of equations, which will play a key role in determining the differential spectrum of $F$. We define a sequence $\{\tau_m\}$ recursively as follows
$$
\tau_1 = \frac{1}{2}, \quad \tau_2 = -\frac{7}{4}, \quad \text{and} \quad \tau_m = \frac{1}{2}\tau_{m-1} - \tau_{m-2} \quad \text{for } m \geq 3.
$$
By solving the auxiliary quadratic equation corresponding to the given
recurrence relation, we obtain the explicit form if $\tau_m$. The general term $ \tau_m $ can be expressed explicitly in closed form as
\begin{equation}\label{taum}
	\tau_m=(\frac{1+\sqrt{-15}}{4})^m+(\frac{1-\sqrt{-15}}{4})^m=\frac{1}{2^{2m-1}}\sum_{i=0}^{[\frac{m}{2}]}\binom{m}{2i}(-15)^{i}.
\end{equation}
For convenience, we provide the following table showing the values of $ \tau_m $ for $ 1 \leq m \leq 10 $.

\begin{table}[!htp]
	\footnotesize
	\centering
	\caption{The values of $\tau_m$ for $1\leqslant m\leqslant 10$ }
	\begin{tabular}{|c|c|c|c|c|c|c|c|c|c|c|}
		\hline
		$m$ & $1$ & $2$ & $3$ & $4$ & $5$ & $6$ & $7$ & $8$ & $9$ & $10$  \\
		\hline
		
		$\tau_m$ & $\frac{1}{2}$ & $-\frac{7}{4}$ & $-\frac{11}{8}$ & $\frac{17}{16}$ & $\frac{61}{32}$ & $-\frac{7}{64}$ & $-\frac{251}{128}$ & $-\frac{223}{256}$ & $\frac{781}{512}$ & $\frac{1673}{1024}$  \\
		\hline

	\end{tabular}
	\label{tab:valuesoflambda}
\end{table}

For a positive integer $r$, let $n_r$ denote the number of solutions $(x_1,x_2,\cdots,x_r)\in(\gf_{q^2}^*)^r$ of the following system of equations:

\begin{equation}
	\Bigg\{ \begin{array}{ll}
		{x_1} + {x_2} + \cdots + {x_r} &= 0\\
		x_1^d + x_2^d + \cdots + x_r^d &= 0.
	\end{array}
\end{equation}
The following lemma was established in \cite{XLZH16}.

\begin{lemma}\label{xiaeqn}  We have,
	$$n_3=
	\begin{cases}
		(q-2)(q^2-1), &\mathrm{if}~ m ~\mathrm{is ~even}, \\
		q(q^2-1), &\mathrm{if}~ m ~\mathrm{is ~odd}
	\end{cases}
	$$
	and 
	$$n_4=
	\begin{cases}
		(5q^2-(\tau_m+6)q+4)(q^2-1), &\mathrm{if}~ m ~\mathrm{is ~even}, \\
		(5q^2-(\tau_m+6)q)(q^2-1), &\mathrm{if}~ m ~\mathrm{is ~odd},
	\end{cases}
	$$
	where $\tau_m$ is given in (\ref{taum}). 
\end{lemma}
Based on Lemma \ref{xiaeqn}, we obtain the following result.
\begin{lemma} Let $N_4$ denote he number of solutions $(x_1,x_2,x_3,x_4)\in(\gf_{q^2})^4$ of the system of equations
	\begin{equation}\label{N4}
		\Bigg\{ \begin{array}{ll}
			{x_1} + {x_2} + x_3 + {x_4} &= 0\\
			x_1^d + x_2^d + x_3^d + x_4^d &= 0.
		\end{array}
	\end{equation}
	Then we have
	$$N_4=
	\begin{cases}
		1+(5q^2-(\tau_m+2)q+2)(q^2-1), \mathrm{if}~ m ~\mathrm{is ~even}, \\
		1+(5q^2-(\tau_m+2)q+6)(q^2-1), \mathrm{if}~ m ~\mathrm{is ~odd},
	\end{cases}
	$$
	where $\tau_m$ is given in (\ref{taum}). 
\end{lemma}
\begin{proof}
	We now consider the solutions of (\ref{N4}) that include zero components. It is clear that $(0,0,0,0)$ is a solution of (\ref{N4}). Moreover, there are no solutions containing exactly three zeros. 
	
	Suppose  $(x_1,x_2,x_3,x_4)$ is a solutions of (\ref{N4}) with exactly two zero components. Without loss of generality, assume $x_3=x_4=0$. Then we must have $x_1+x_2=0$, which implies $x_1=x_2\in\gf_{q^2}^*$.
	By considering all possible positions of the two zero components, we find that there are $\binom{4}{2} = 6$ such configurations, each contributing $q^2 - 1$ solutions. Thus, the total number of solutions with exactly two zero components is $6(q^2 - 1)$.
	
	Similarly, the number of solutions with exactly one zero component is $4n_3$, and the number of solutions with no zero components is $n_4$, where $n_3$ and $n_4$ are given in Lemma \ref{xiaeqn}. Therefore, the total number of solutions to \eqref{N4} is
	\[N_4=1+6(q^2-1)+4n_3+n_4\]
	and substituting the values of $n_3$ and $n_4$ from Lemma~\ref{xiaeqn} yields the desired result.
\end{proof}

\section{The differential properties of $F$}\label{sec:main}
Let $F(x)=x^{3q-2}$ be a power function over $\gf_{q^2}$, where $q=2^m$ and $m$ is a positive even number. 
In this section, we determine the differential uniformity and the differential spectrum of $F$. 

Recall that for each $b\in\mathbb{F}_{q^2}$,  $ \delta(1,b) $ is defined as
$$\delta(1,b)=|\{x\in\gf_{q^2}:F(x+1)+F(x)=b\}|.$$ We begin with the following result.
\begin{theorem}\label{delta1} We have $\delta(1,1)=q$.
\end{theorem}

\begin{proof} It sufficient to count the number of solutions of the equation
	\begin{equation}\label{eqndelta1}
		(x+1)^d+x^d=1.
	\end{equation}
	Clearly, $x=0$ and $x=1$ are solutions of (\ref{eqndelta1}). Now suppose $x\neq 0, 1$. Multiplying both sides of the equation by $ x^2(x+1)^2 $, we obtain
	\[(x+1)^{3q}x^2+x^{3q}(x+1)^2=x^2(x+1)^2.\]
	This can be rewritten as
	\[(x^q+x)^2(x^q+x^2)=0.\]
	We consider two cases. If $x^q+x=0$, then $x\in\gf_q$. It is easy to verify that every $x\in\gf_q$ satisfies (\ref{eqndelta1}). If $x^q+x^2=0$, then $x\in\gf_{\frac{q}{2}}$. However, since $ m $ is even, we have $\gf_{\frac{q}{2}}\cap\gf_{q^2}=\gf_2$. Therefore, only $ x = 0 $ and $ x = 1 $ satisfy this condition. 	Thus, all solutions of \eqref{eqndelta1} are in $ \mathbb{F}_q $, and hence $ \delta(1,1) = q $, as desired.
\end{proof}

Based on Lemma \ref{xyz} and Theorem \ref{image}, we are now ready to state the main result of this section.

\begin{theorem}\label{deltab}
	For $b\in\gf_{q^2}\setminus\{1\}$, we have $\delta(1,b)\leq 4$.
\end{theorem}
\begin{proof}
	We consider the equation 
	\begin{equation}\label{eqndeltab}
		(x+1)^d+x^d=b
	\end{equation}
	for $b\in\gf_{q^2}\setminus\{1\}$. Let $x\in\gf_{q^2}\setminus\gf_2$ be a solution of (\ref{eqndeltab}). By Lemma \ref{xyz}, we write
	$$x + 1 = y_1 z_1 \quad \text{and} \quad x = y_2 z_2,$$
	where $y_1,y_2\in\gf_q^*$ and $z_1,z_2\in U_{q+1}$. Adding these two expressions gives
	\[y_1z_1+y_2z_2=1.\]
	Taking the $q$-th power of the above equation, we obtain
	\[y_1z_1^{-1}+y_2z_2^{-1}=1.\]
	Thus, we arrive at the following system of linear equations
	$$
	\begin{bmatrix}
		z_1 & z_2 \\
		z_1^{-1} & z_2^{-1}
	\end{bmatrix}
	\begin{bmatrix}
		y_1 \\
		y_2
	\end{bmatrix}
	=
	\begin{bmatrix}
		1 \\
		1
	\end{bmatrix}
	.$$
	This completes the setup for analyzing the number of solutions to \eqref{eqndeltab}, which will lead to the bound $ \delta(1,b) \leq 4 $.
	
	The determinant of the coefficient matrix is zero if and only if $z_1=z_2$. In that case, we have $$(y_1+y_2)z_1=(y_1+y_2)z_1^{-1}=1,$$
	which implies $ z_1 = 1 $ and $ y_1 + y_2 = 1 $. Then
	$$
	b = (y_1 z_1)^d + (y_2 z_2)^d = y_1 + y_2 = 1,
	$$
	a contradiction to our assumption that $ b \ne 1 $. Therefore, in what follows, we may assume $ z_1 \ne z_2 $.
	Under this assumption, the system of linear equations has a unique solution given by
	$$
	\begin{cases}
		y_1=\frac{z_1(z_2+1)^2}{(z_1+z_2)^2} \\
		y_2=\frac{z_2(z_1+1)^2}{(z_1+z_2)^2}.
	\end{cases}
	$$
	If $x$ is a solution of (\ref{eqndeltab}), then  we can express $ b $ as
	\[b=y_1z_1^{-5}+y_2z_2^{-5}=(z_1^{-2}+z^{-1}_1z^{-1}_2+z_2^{-2}+z_1^{-2}z_2^{-1}+z_1^{-1}z_2^{-2})^2.\]
	Note that $z_1$ and $z_2$ are uniquely determined by $x$. We now aim to determine the maximum number of distinct pairs $ (z_1, z_2) \in U_{q+1}^2 $ that correspond to a fixed value of $ b $.
	
	Let $ u_1 = z_1^{-1} $ and $ u_2 = z_2^{-1} $, then $u_1\neq u_2$. Moreover, the expression for $ b $ becomes
	$$
	b = \left( u_1^2 + u_1 u_2 + u_2^2 + u_1^2 u_2 + u_1 u_2^2 \right)^2.
	$$
	Since $u_1,u_2\in U_{q+1}$, we have the identity $u_1u_2=(u_1+u_2)^{1-q}$. Let $t=u_1+u_2$. Substituting this into the expression for $ b $, we obtain
	\begin{equation}\label{bt}
		b=(t^2+t^{1-q}+t^{2-q})^2, 
	\end{equation}
	which is a function of a single variable $t$. Now, since $b\neq 1$, it follows from Theorem \ref{image} that there are at most two values of $t$ satisfying (\ref{bt}) for a given $b$. By Lemma \ref{usum}, if $t=u_1+u_2=u_3+u_4$ for some $u_3,u_4\in U_{q+1}$, then $(u_3,u_4)=(u_1,u_2)$ or $(u_3,u_4)=(u_2,u_1)$. Recall
	\[x=y_2z_2=\frac{z^2_2(z_1+1)^2}{(z_1+z_2)^2}=\big(\frac{u_2^{-1}(u_1^{-1}+1)}{u_1^{-1}+u_2^{-1}}\big)^2=\big(\frac{u_1+1}{u_1+u_2}\big)^2.\]
	Each pair $(u_1,u_2)$ gives rise to at most one value of $ x $. Since there are at most two possible values of $ t = u_1 + u_2 $, and each such $ t $ corresponds to at most two distinct ordered pairs $ (u_1, u_2) $, we conclude that \eqref{eqndeltab} has at most four solutions in total. Therefore, $ \delta(1,b) \leq 4 $, which completes the proof.
\end{proof}

Now we are in the position to determine the differential spectrum of $F$. The main result of this paper is as follows.

\begin{theorem}\label{DS} Let $m\geq 4$ be an even integer and $q=2^m$. Define the power function $ F(x) = x^{3q - 2} $ over $ \mathbb{F}_{q^2} $. Then the differential uniformity of $F$ is equal to $q$, and the differential spectrum of $F$ is given by

\begin{equation*}
DS_F=[\omega_0=\frac{1}{8}(5q^2+(-\tau_m+4)q-7),\omega_2=\frac{1}{4}(q^2+(\tau_m-2)q-1),
\end{equation*}	
\begin{equation*}
\omega_4=\frac{1}{8}(q^2-\tau_mq+1),\omega_q=1],
\end{equation*}	
	where $\tau_m$ is defined in (\ref{taum}). Furthermore, $F$ is locally differentially $4$-uniform.
\end{theorem}
\begin{proof} 
	Note that the solutions of $(x+1)^d+x^d=b$ come in pairs, hence, $\omega_i=0$ for all odd $i$. By Theorems \ref{delta1} and \ref{deltab}, we conclude that the possible nonzero entries in the differential spectrum of $F$ are $\omega_0,\omega_2,\omega_4$ and $\omega_q$. Since $m\geq 4$, we have $q>4$. It follows from Theorems \ref{delta1} and \ref{deltab} that $F$ is differentially $q$-uniform and that $\omega_q=1$. Moreover, by (\ref{omegaiomega}) and Lemma \ref{i2omega}, the multiplicities $ \omega_0, \omega_2, \omega_4 $, and $ \omega_q $ satisfy the following system of equations:
	\begin{equation*}
		\Bigg\{ \begin{array}{ll}
			\omega_0+\omega_2+\omega_4+\omega_q&=q^2\\
			2\omega_2+4\omega_4+q\omega_q&=q^2\\
			4\omega_2+16\omega_4+q^2\omega_q&=4q^2-(\tau_m+2)q+1\\
			\omega_q&=1,\\
		\end{array}
	\end{equation*}
	where $\tau_m$ is defined in (\ref{taum}). Solving this system yields the explicit values of $ \omega_0, \omega_2, \omega_4 $, and $ \omega_q $, which constitute the differential spectrum of $ F $. Furthermore, it can be shown that $|\tau_m|<2^m$. Therefore, when $ m \geq 4 $, we have $ \omega_4 > 0 $, which implies that there exist values $ b \in \mathbb{F}_{q^2} $ such that $ \delta(1,b) = 4 $. Hence $F$ is locally differentially $4$-uniform.
\end{proof}

\begin{remark}
	When $ m = 2 $, we have $ q = 4 $, and the function becomes $ F(x) = x^{10} $, defined over $ \mathbb{F}_{16} $. In this case, the differential spectrum of $ F $ is given by
	$$
	\mathrm{DS}_F = [\omega_0 = 12, \omega_4 = 4],
	$$
	which can also be derived by combining (\ref{omegaiomega}), Lemma \ref{i^2omega}, and Theorems~\ref{delta1} and \ref{deltab}.
\end{remark}

In the following, we provide some numerical examples to illustrate and verify our theoretical results.

\begin{example}
	Let $ m = 4 $. Then $ q = 2^4 = 16 $, and the function becomes $ F(x) = x^{46} $ over $ \mathbb{F}_{256} $. According to Theorem~\ref{DS}, the differential spectrum of $ F $ is
	$$
	\mathrm{DS}_F =	[\omega_0 = 165,\ \omega_2 = 60,\ \omega_4 = 30,\ \omega_{16} = 1],
	$$
	which coincides with the result obtained by direct computation using MAGMA.
\end{example}

\begin{example}
	Let $ m = 6 $. Then $ q = 2^6 = 64 $, and the function becomes $ F(x) = x^{190} $ over $ \mathbb{F}_{4096} $. By Theorem~\ref{DS}, the differential spectrum of $ F $ is
	$$
	\mathrm{DS}_F =	[\omega_0 = 2592,\ \omega_2 = 990,\ \omega_4 = 513,\ \omega_{64} = 1],
	$$
	which coincides with the result obtained by direct computation using MAGMA.
\end{example}

\begin{example}
	Let $ m = 8 $. Then $ q = 2^8 = 256 $, and the function becomes $ F(x) = x^{766} $ over $ \mathbb{F}_{65536} $. Applying Theorem~\ref{DS}, the differential spectrum of $ F $ is
	$$
	\mathrm{DS}_F =	[\omega_0 = 41115,\ \omega_2 = 16200,\ \omega_4 = 8220,\ \omega_{256} = 1],
	$$
	which coincides with the result obtained by direct computation using MAGMA.
\end{example}

\section{Concluding Remarks}\label{sec:con}

In this paper, we investigated the differential spectrum of the power function $ F(x) = x^{3q - 2} $ over $ \mathbb{F}_{q^2} $, where $ q = 2^m $ and $ m\geq 4 $ is an even integer. Here, $ 3q - 2 $ is a Niho exponent. However, the differential spectrum of $ x^{3q - 2} $ over $ \mathbb{F}_{q^2} $ for other values of $ q $, such as when $ q $ is an odd prime power, remains an open problem. 

The differential spectrum of a power function played a crucial role in analyzing the cross-correlation distribution between an $ m $-sequence and its decimation. These findings provide a foundation for further research in this area. Future work will focus on extending our results to broader parameter ranges and exploring the differential properties of other power functions with Niho exponents.

\bibliographystyle{IEEEtranS}

\bibliography{Y}

@article {A21,
    AUTHOR = {Abdukhalikov, K.},
     TITLE = {Equivalence classes of {N}iho bent functions},
   JOURNAL = {Des. Codes Cryptogr.},
  FJOURNAL = {Designs, Codes and Cryptography. An International Journal},
    VOLUME = {89},
      YEAR = {2021},
    NUMBER = {7},
     PAGES = {1509--1534},
}

@article{li2023differential,
  title={Differential spectra of a class of power permutations with Niho exponents},
  author={Li, Zhen and Yan, Haode},
  journal={Advances in Mathematics of Communications},
  volume={17},
  number={6},
  pages={1468--1475},
  year={2023},
  publisher={Advances in Mathematics of Communications}
}

@article{golouglu2023exponential,
  title={An exponential bound on the number of non-isotopic commutative semifields},
  author={G{\"o}lo{\u{g}}lu, Faruk and K{\"o}lsch, Lukas},
  journal={Transactions of the American Mathematical Society},
  volume={376},
  number={03},
  pages={1683--1716},
  year={2023}
}

@book{carlet2021boolean,
  title={Boolean functions for cryptography and coding theory},
  author={Carlet, Claude},
  year={2021},
  publisher={Cambridge University Press}
}

@article{li2023two,
  title={Two new infinite families of APN functions in trivariate form},
  author={Li, Kangquan and Kaleyski, Nikolay},
  journal={IEEE Transactions on Information Theory},
  volume={70},
  number={2},
  pages={1436--1452},
  year={2024},
  publisher={IEEE}
}

@article{golouglu2022biprojective,
  title={Biprojective almost perfect nonlinear functions},
  author={G{\"o}lo{\u{g}}lu, Faruk},
  journal={IEEE Transactions on Information Theory},
  volume={68},
  number={7},
  pages={4750--4760},
  year={2022},
  publisher={IEEE}
}

@article{BS91,
  title={Differential cryptanalysis of \uppercase{D}\uppercase{E}\uppercase{S}-like cryptosystems},
  author={E. Biham and A. Shamir},
  journal={J. Cryptol.},
  volume={4},
  number={1},
  pages={3-72},
  year={1991},
}

@article{BCC10,
  title={Differential properties of power functions},
  author={C. Blondeau and A. Canteaut and P. Charpin},
  journal={Int. J. Inf. Coding Theory},
  volume={1},
  number={2},
  pages={149-170},
  year={2010},
}

@article{BCC11,
  title={Differential properties of $x\rightarrow x^{2^t-1}$},
  author={C. Blondeau and A. Canteaut and P. Charpin},
  journal={IEEE Trans. Inform. Theory},
  volume={57},
  number={12},
  pages={8127-8137},
  year={2011},
}

@article{BP14,
  title={More differentially 6-uniform power functions},
  author={C. Blondeau and L. Perrin},
  journal={Des. Codes Cryptogr.},
  volume={73},
  number={2},
  pages={487-505},
  year={2014},
}

@article {BBM08,
    AUTHOR = {Bracken, Carl and Byrne, Eimear and Markin, Nadya and McGuire,
              Gary},
     TITLE = {New families of quadratic almost perfect nonlinear trinomials
              and multinomials},
   JOURNAL = {Finite Fields Appl.},
  FJOURNAL = {Finite Fields and their Applications},
    VOLUME = {14},
      YEAR = {2008},
    NUMBER = {3},
     PAGES = {703--714},
      ISSN = {1071-5797},
   MRCLASS = {11T06 (11T71 94A60)},
  MRNUMBER = {2435056},
MRREVIEWER = {S. D. Cohen},
}

@article {BCCCV20,
    AUTHOR = {Budaghyan, Lilya and Calderini, Marco and Carlet, Claude and
              Coulter, Robert S. and Villa, Irene},
     TITLE = {Constructing {APN} functions through isotopic shifts},
   JOURNAL = {IEEE Trans. Inform. Theory},
  FJOURNAL = {Institute of Electrical and Electronics Engineers.
              Transactions on Information Theory},
    VOLUME = {66},
      YEAR = {2020},
    NUMBER = {8},
     PAGES = {5299--5309},
      ISSN = {0018-9448},
   MRCLASS = {94A60 (11T71 94D10)},
  MRNUMBER = {4130674},
MRREVIEWER = {Kamil Otal},
 
}

@article {BHK20,
    AUTHOR = {Budaghyan, Lilya and Helleseth, Tor and Kaleyski, Nikolay},
     TITLE = {A new family of {APN} quadrinomials},
   JOURNAL = {IEEE Trans. Inform. Theory},
  FJOURNAL = {Institute of Electrical and Electronics Engineers.
              Transactions on Information Theory},
    VOLUME = {66},
      YEAR = {2020},
    NUMBER = {11},
     PAGES = {7081--7087},
      ISSN = {0018-9448},
   MRCLASS = {94D10},
  MRNUMBER = {4173628},
MRREVIEWER = {Aleksei Udovenko},

}

@article{CHNC13,
  title={Differential spectrum of some power functions in odd prime characteristic},
  author={S.-T. Choi and S. Hong and J.-S. No and H. Chung},
  journal={Finite Fields Appl.},
  volume={21},
  number={},
  pages={11-29},
  year={2013},
}

@article {D991,
    AUTHOR = {Dobbertin, Hans},
     TITLE = {Almost perfect nonlinear power functions on {${\rm GF}(2^n)$}:
              the {W}elch case},
   JOURNAL = {IEEE Trans. Inform. Theory},
  FJOURNAL = {Institute of Electrical and Electronics Engineers.
              Transactions on Information Theory},
    VOLUME = {45},
      YEAR = {1999},
    NUMBER = {4},
     PAGES = {1271--1275},
      ISSN = {0018-9448},
   MRCLASS = {94A55 (94A60)},
  MRNUMBER = {1686267},

}

@article {D992,
    AUTHOR = {Dobbertin, Hans},
     TITLE = {Almost perfect nonlinear power functions on {${\rm GF}(2^n)$}:
              the {N}iho case},
   JOURNAL = {Inform. and Comput.},
  FJOURNAL = {Information and Computation},
    VOLUME = {151},
      YEAR = {1999},
    NUMBER = {1-2},
     PAGES = {57--72},
      ISSN = {0890-5401},
   MRCLASS = {94A60 (68P25)},
  MRNUMBER = {1692816},

}

@incollection {D01,
    AUTHOR = {Dobbertin, Hans},
     TITLE = {Almost perfect nonlinear power functions on {${\rm GF}(2^n)$}:
              a new case for {$n$} divisible by {$5$}},
 BOOKTITLE = {Finite fields and applications ({A}ugsburg, 1999)},
     PAGES = {113--121},
 PUBLISHER = {Springer, Berlin},
      YEAR = {2001},
   MRCLASS = {11T71 (94A60)},
  MRNUMBER = {1849084},
}

@article{DHKM01,
  title={Ternary m-sequences with three-valued cross-correlation function: new decimations of Welch and Niho type},
  author={H. Dobbertin and T. Helleseth and P. V. Kumar and H. Martinsen},
  journal={IEEE Trans. Inform. Theory},
  volume={47},
  number={4},
  pages={1473-1481},
  year={2001},
}

@article{DMMPW03,
  title={{APN} functions in odd characteristic},
  author={H. Dobbertin and D. Mills and E. N. M\"{u}ller and A. Pott and W. Willems},
  journal={Discrete Math.},
  volume={267},
  number={1-3},
  pages={95-112},
  year={2003},
}

@article{HLXZT23,
  title={The differential spectrum and boomerang spectrum of a class of locally-{APN} functions},
  author={Z. Hu and N. Li and L. Xu and X. Zeng and X. Tang},
  journal={Des. Codes Cryptogr.},
  volume={91},
  number={5},
  pages={1695-1711},
  year={2023},
}

@article{HRS99,
  title={New families of almost perfect nonlinear power functions},
  author={T. Helleseth and C. Rong and D. Sandberg},
  journal={IEEE Trans. Inform. Theory},
  volume={45},
  number={2},
  pages={475-485},
  year={1999},
}

@article{JLLQ21,
  title={Differential spectrum of a class of power functions},
  author={S. Jiang and K. Li and Y. Li and L. Qu},
  journal={J. Cryptology},
  volume={9},
  number={3},
  pages={484-495},
  year={2021},
}

@article{JLLQ22,
  title={Differential and boomerang spectrums of some power permutations},
  author={S. Jiang and K. Li and Y. Li and L. Qu},
  journal={Cryptogr. Commun.},
  volume={14},
  number={},
  pages={371-393},
  year={2022},
}

@article {LZ19,
    AUTHOR = {Li, Nian and Zeng, Xiangyong},
     TITLE = {A survey on the applications of {N}iho exponents},
   JOURNAL = {Cryptogr. Commun.},
  FJOURNAL = {Cryptography and Communications. Discrete Structures, Boolean
              Functions and Sequences},
    VOLUME = {11},
      YEAR = {2019},
    NUMBER = {3},
     PAGES = {509--548},
      ISSN = {1936-2447},
   MRCLASS = {12Y05 (11T06 12E20 94A55 94A60 94B05)},
  MRNUMBER = {3946534},

}

@article{MXLH22,
  title={On the differential properties of the power function $x^{p^m+2}$},
  author={Y. Man and Y. Xia and C. Li and T. Helleseth},
  journal={Finite Fields Appl.},
  volume={84},
  number={10},
  pages={1-22},
  year={2022},
}

@book{N72,
  title={Multivalued cross-correlation functions between two maximal linear recursive sequences},
  author={Y. Niho},
  publisher={Ph.D. dissertation, Dept. Elect. Eng., Univ. Southern California},
  country={Log Angeles, CA, USA},
  year={1972},
}

@article{N94,
  title={Differentially uniform mappings for cryptography},
  author={K. Nyberg},
  journal={Advances in Cryptology-EUROCRYPT 1994},
  volume={765},
  number={},
  pages={55-64},
  year={1994},
}

@article{PLZ23,
  title={On the differential spectrum of a differentially 3-uniform power function},
  author={T. Pang and N. Li and X. Zeng},
  journal={Finite Fields Appl.},
  volume={87},
  number={},
  pages={},
  year={2023},
}

@article{TY23,
  title={Differential spectrum of a class of {APN} power functions},
  author={X. Tan and H. Yan},
  journal={Des. Codes Cryptogr.},
  volume={91},
  number={},
  pages={2755-2768},
  year={2023},
}

@article {TD21,
    AUTHOR = {Tang, Chunming and Ding, Cunsheng},
     TITLE = {An infinite family of linear codes supporting 4-designs},
   JOURNAL = {IEEE Trans. Inform. Theory},
  FJOURNAL = {Institute of Electrical and Electronics Engineers.
              Transactions on Information Theory},
    VOLUME = {67},
      YEAR = {2021},
    NUMBER = {1},
     PAGES = {244--254},
      ISSN = {0018-9448},
   MRCLASS = {94B15 (94B05)},
  MRNUMBER = {4231952},

}

@article {TC17,
    AUTHOR = {S. Tian and Y. Chen},
     TITLE = {Differential spectra for a class of power functions over
              {$\Bbb F_{p^n}$}},
   JOURNAL = {J. Systems Sci. Math. Sci.},
  FJOURNAL = {Journal of Systems Science and Mathematical Sciences. Xitong
              Kexue yu Shuxue},
    VOLUME = {37},
      YEAR = {2017},
    NUMBER = {5},
     PAGES = {1351--1367},
      ISSN = {1000-0577},
   MRCLASS = {94A60 (11T71)},
  MRNUMBER = {3711734},
}

@article{TLWZTJ23,
  title={On the differential spectrum and the APcN property of a class of power functions over finite fields},
  author={Z. Tu and N. Li and Y. Wu and X. Zeng and X. Tang and Y. Jiang},
  journal={IEEE Trans. Inform. Theory},
  volume={69},
  number={1},
  pages={582-597},
  year={2023},
}

@article {XLZH16,
    AUTHOR = {Xia, Yongbo and Li, Nian and Zeng, Xiangyong and Helleseth,
              Tor},
     TITLE = {An open problem on the distribution of a {N}iho-type
              cross-correlation function},
   JOURNAL = {IEEE Trans. Inform. Theory},
  FJOURNAL = {Institute of Electrical and Electronics Engineers.
              Transactions on Information Theory},
    VOLUME = {62},
      YEAR = {2016},
    NUMBER = {12},
     PAGES = {7546--7554},
      ISSN = {0018-9448},
   MRCLASS = {94A55},
  MRNUMBER = {3578242},
 
}

@article{XYY18,
  title={On a conjecture of differentially 8-uniform power functions},
  author={M. Xiong and H. Yan and P. Yuan},
  journal={Des. Codes Cryptogr.},
  volume={86},
  number={8},
  pages={1601-1621},
  year={2018},
}

@article{XZLH20,
  title={The differential spectrum of a ternary power function},
  author={Y. Xia and X. Zhang and C. Li and T. Helleseth},
  journal={Finite Fields Appl.},
  volume={64},
  number={},
  pages={1-16},
  year={2020},
}

@article{XML23,
  title={On the Niho type locally-{APN} power functions and their boomerang spectrum},
  author={X. Xie and S. Mesnager and N. Li and D. He and X. Zeng},
  journal={IEEE Trans. Inform. Theory},
  volume={69},
  number={6},
  pages={4056-4064},
  year={2023},
}

@article{XY17,
  title={A note on the differential spectrum of a differentially 4-uniform power function},
  author={M. Xiong and H. Yan},
  journal={Finite Fields Appl.},
  volume={48},
  number={},
  pages={117-125},
  year={2017},
}

@article{YL21,
  title={Differential spectra of a class of power permutations with characteristic 5},
  author={H. Yan and C. Li},
  journal={Des. Codes Cryptogr.},
  volume={89},
  number={6},
  pages={1181-1191},
  year={2021},
}

@article{YMT23,
  title={The complete differential spectrum of a class of power permutations over odd characteristic finite fields},
  author={H. Yan and S. Mesnager and X. Tan},
  journal={IEEE Trans. Inform. Theory},
  volume={69},
  number={11},
  pages={7426-7438},
  year={2023},
}

@article{YMT24,
  title={On a class of {APN} power functions over odd characteristic finite fields: Their differential spectrum and c-differential properties},
  author={H. Yan and S. Mesnager and X. Tan},
  journal={Discrete Math.},
  volume={347},
  number={4},
  pages={},
  year={2024},
}

@article{YX22,
  title={The differential spectrum of the power mapping $x^{p^n-3}$},
  author={H. Yan and Y. Xia and C. Li and T. Helleseth and M. Xiong and J. Luo},
  journal={IEEE Trans. Inform. Theory},
  volume={68},
  number={8},
  pages={5535-5547},
  year={2022},
}

@article {YDZQ25,
    AUTHOR = {Yuan, Wenping and Du, Xiaoni and Zhou, Huan and Qiao, Xingbin},
     TITLE = {The differential uniformity of the power functions
              {$x^{\frac{p^n + 5}{2}}$} over {$\Bbb{F}_{p^n}$}},
   JOURNAL = {Finite Fields Appl.},
  FJOURNAL = {Finite Fields and their Applications},
    VOLUME = {105},
      YEAR = {2025},
     PAGES = {Paper No. 102622, 15},
      ISSN = {1071-5797},
   MRCLASS = {11T06 (11T71)},
  MRNUMBER = {4884211},

}

@article {ZKW09,
    AUTHOR = {Zha, Zhengbang and Kyureghyan, Gohar M. and Wang, Xueli},
     TITLE = {Perfect nonlinear binomials and their semifields},
   JOURNAL = {Finite Fields Appl.},
  FJOURNAL = {Finite Fields and their Applications},
    VOLUME = {15},
      YEAR = {2009},
    NUMBER = {2},
     PAGES = {125--133},
      ISSN = {1071-5797},
   MRCLASS = {12E20 (11T06 12K10 51E15 94A60)},
  MRNUMBER = {2494329},
MRREVIEWER = {Dieter Jungnickel},

}

@article {ZW09,
    AUTHOR = {Zha, Zhengbang and Wang, Xueli},
     TITLE = {New families of perfect nonlinear polynomial functions},
   JOURNAL = {J. Algebra},
  FJOURNAL = {Journal of Algebra},
    VOLUME = {322},
      YEAR = {2009},
    NUMBER = {11},
     PAGES = {3912--3918},
      ISSN = {0021-8693},
   MRCLASS = {12E20},
  MRNUMBER = {2556130},
MRREVIEWER = {David Brink},

}

@article{ZW11,
  title={Almost perfect nonlinear power functions in odd characteristic},
  author={Z. Zha and X. Wang},
  journal={IEEE Trans. Inform. Theory},
  volume={57},
  number={7},
  pages={4826-4832},
  year={2011},
}

\end{document}